\documentclass[conference]{IEEEtran}
\usepackage{graphicx}
\usepackage{dsfont}
\usepackage{cite}
\usepackage{color}
\usepackage{amsmath,amssymb,amsthm}
\newcommand\norm[1]{\left\lVert#1\right\rVert}
\usepackage{algorithm,algorithmic,multicol}

\newtheoremstyle{mystyle}
{}
{}
{\itshape}
{}
{\bfseries}
{.}
{ }
{}
\theoremstyle{mystyle}

\newtheorem{proposition}{Proposition}

\begin{document}
\title{Energy Efficient ADC Bit Allocation and Hybrid Combining for Millimeter Wave MIMO Systems}

\author{\IEEEauthorblockN{Aryan Kaushik$^1$, Christos Tsinos$^2$, Evangelos Vlachos$^1$, John Thompson$^1$}
\IEEEauthorblockA{$^1$Institute for Digital Communications, The University of Edinburgh, United Kingdom.\\
$^2$Interdisciplinary Centre for Security, Reliability and Trust, University of Luxembourg, Luxembourg.\\
Emails: \{a.kaushik, e.vlachos, j.s.thompson\}@ed.ac.uk, christos.tsinos@uni.lu}}

\maketitle

\begin{abstract}
Low resolution analog-to-digital converters (ADCs) can be employed to improve the energy efficiency (EE) of a wireless receiver since the power consumption of each ADC is exponentially related to its sampling resolution and the hardware complexity. In this paper, we aim to jointly optimize the sampling resolution, i.e., the number of ADC bits, and analog/digital hybrid combiner matrices which provides highly energy efficient solutions for millimeter wave multiple-input multiple-output systems. A novel decomposition of the hybrid combiner to three parts is introduced: the analog combiner matrix, the bit resolution matrix and the baseband combiner matrix. The unknown matrices are computed as the solution to a matrix factorization problem where the optimal, fully digital combiner is approximated by the product of these matrices. An efficient solution based on the alternating direction method of multipliers is proposed to solve this problem. The simulation results show that the proposed solution achieves high EE performance when compared with existing benchmark techniques that use fixed ADC resolutions.
\end{abstract}


%
\IEEEpeerreviewmaketitle

\section{Introduction}
The analog/digital (A/D) hybrid beamforming architectures for millimeter wave (mmWave) multiple-input multiple-output (MIMO) systems reduce the hardware complexity and the power consumption through fewer radio frequency (RF) chains and support multi-stream communication with good capacity performance \cite{ayachTWC2014, aryanIET2016, hanCM2015}. Designing such systems for high energy efficiency (EE) gains would leverage their significance \cite{ranziSAC2016, tsinosSAC2017}. An alternative solution to reduce the power consumption and hardware complexity is by reducing the resolution sampling \cite{heathSSP2016}. Some approaches have been applied in hybrid mmWave MIMO systems for EE maximization and low complexity with full resolution \cite{aryanTGCN2019} and low resolution \cite{aryanICC2019}.

The existing literature mostly discusses full or high resolution analog-to-digital converters (ADCs) with a small number of RF chains or low resolution ADCs with a large number of RF chains: either way only the fixed resolution ADCs are taken into account. References \cite{ranziSAC2016, tsinosSAC2017} consider EE optimization problems for A/D hybrid transceivers but with fixed and high resolution digital-to-analog converters (DACs)/ADCs. Reference \cite{aryanICC2019} proposes a novel EE maximization transmission technique with subset selection optimization to find the best subset of the active RF chains and DAC resolution, which can be extended to low resolution ADCs at the receiver (RX). 
Reference \cite{jmoITG2016} suggests implementing fixed and low resolution ADCs with few RF chains. Reference \cite{zhangSAC2017} studies the idea of a mixed-ADC architecture where a better energy-rate trade off is achieved by using mixed resolution ADCs but still with a fixed resolution for each ADC and it does not consider A/D hybrid beamforming. A hybrid beamforming system with fixed and low resolution ADCs has been analyzed for channel estimation in \cite{aryanEUSIPCO2018}. Varying resolution ADCs can be implemented at the RX \cite{zhangTWC2016} which may provide a better solution than fixed and low resolution ADCs. Extra care is needed when deciding the range of number of ADC bits as the total ADC power consumption can be dominated by only a few high resolution ADCs. Thus, a good trade-off between power consumption and performance is to consider the range of 1-8 bits for the varying number of ADC bits.


\subsubsection*{Contributions}
This paper designs an optimal EE solution for a mmWave A/D hybrid receiver MIMO system by introducing the novel decomposition of the A/D hybrid combiner to three parts representing the analog combiner matrix, the bit resolution matrix and baseband combiner matrix. Our aim is to minimize the distance between this decomposition, which is expressed as the product of three matrices, and the fully digital combiner matrix. The joint problem is decomposed into a series of sub-problems which are solved using an alternating optimization framework, i.e., alternating direction method of multipliers (ADMM) is developed to obtain the unknown matrices. The proposed design has high flexibility, given that the analog combiner is codebook-free, thus there is no restriction on the angular vectors and different bit resolutions can be assigned to each ADC. Our proposed solution optimizes the resolution on a packet-by-packet basis for each one of the ADCs unlike existing approaches that are based on fixed resolution sampling. We also implement an exhaustive search approach \cite{ranziSAC2016} for comparison which provides the upper bound for EE maximization.

\emph{Notation:} $\textbf{A}$, $\textbf{a}$ and $\textit{a}$ denote a matrix, a vector and a scalar, respectively. The complex conjugate transpose and transpose of $\textbf{A}$ are denoted as $\textbf{A}^{H}$ and $\textbf{A}^{T}$; $|a|$ represents the determinant of $a$;
$\textbf{I}_{N}$ represents $N \times N$ identity matrix; $\textbf{X} \in \mathbb{C}^{A \times B}$ and $\textbf{X} \in \mathbb{R}^{A \times B}$ denote $A \times B$ size $\mathbf{X}$ matrix with complex and real entries, respectively; $\mathcal{C}\mathcal{N} (\textbf{a}, \textbf{A})$ denotes a complex Gaussian vector having mean $\textbf{a}$ and covariance matrix $\textbf{A}$; 
$[\mathbf{A}]_{kl}$ is the matrix entry at the $k$-th row and $l$-th column. The indicator function $\mathds{1}_{\mathcal{S}}\left\{\mathbf{A}\right\}$ of a set $\mathcal{S}$ that acts over a matrix $\mathbf{A}$  is defined as $0 \hspace{1mm} \forall \hspace{1mm} \mathbf{A} \in \mathcal{S}$ and $\infty \hspace{1mm} \forall \hspace{1mm} \mathbf{A} \notin \mathcal{S}$. 

\section{A/D Hybrid MmWave MIMO System}
\subsection{MmWave Channel Model}
MmWave channels can be modeled by a narrowband clustered channel model due to different channel settings such as number of multipaths, amplitudes, etc., with $N_{\textrm{cl}}$ clusters and $N_{\textrm{ray}}$ propagation paths in each cluster \cite{ayachTWC2014}. Considering a single user mmWave system with $N_{\textrm{T}}$ antennas at the transmitter (TX), transmitting $N_\textrm{s}$ data streams to $N_\textrm{R}$ antennas at the RX, the mmWave channel matrix can be written as follows:
\begin{equation}\label{eq:channel_model}
\mathbf{H} = \sqrt{\frac{N_\textrm{T}N_\textrm{R}}{N_\textrm{cl} N_\textrm{ray}}} \sum_{i=1}^{N_{\textrm{cl}}} \sum_{l=1}^{N_{\textrm{ray}}} \alpha_{il} \mathbf{a}_{\textrm{R}}(\phi_{il}^{r}) \mathbf{a}_{\textrm{T}}(\phi_{il}^{t})^H,
\end{equation}
where $\alpha_{il} \in \mathcal{C}\mathcal{N}(0,\sigma_{\alpha,i}^2)$ is the gain term with $\sigma_{\alpha,i}^2$ being the average power of the $i^{th}$ cluster. Furthermore, $\mathbf{a}_{\textrm{T}}(\phi_{il}^{t})$ and $\mathbf{a}_{\textrm{R}}(\phi_{il}^{r})$ represent the normalized transmit and receive array response vectors \cite{ayachTWC2014}, where $\phi_{il}^{t}$ and $\phi_{il}^{r}$ denote the azimuth angles of departure and arrival, respectively. We use uniform linear array (ULA) antennas for simplicity and model the antenna elements at the RX as ideal sectored elements \cite{singh}. However, the proposed technique is not limited to this setup and can be easily extended to the case of wideband channels and uniform planar/circular arrays. 

\subsection{A/D Hybrid MIMO System Model}
Based on the A/D hybrid beamforming scheme in the large-scale mmWave MIMO communication systems, the number of RX RF chains $L_{\textrm{R}}$ follows the limitation $N_\textrm{s} \leq L_\textrm{R} \leq N_\textrm{R}$ \cite{ayachTWC2014, aryanIET2016}. The matrices $\mathbf{W}_{\textrm{RF}} \in \mathbb{C}^{N_\textrm{R}\times L_\textrm{R}}$ and $\mathbf{W}_{\textrm{BB}} \in \mathbb{C}^{L_\textrm{R}\times N_\textrm{s}}$ denote the analog combiner and baseband (or digital) combiner matrices, respectively. The analog combiner matrix $\mathbf{W}_{\textrm{RF}}$ is based on phase shifters, i.e., the elements that have unit modulus and continuous phase. Thus, $\mathbf{W}_\textrm{RF} \in \mathcal{W}^{N_\textrm{R} \times L_\textrm{R}}$ where the set $\mathcal{W}$ represents the set of possible phase shifts in $\mathbf{W}_\textrm{RF}$ and for a variable $a$, is defined as, $\mathcal{W} = \left\{a \in \mathbb{C} \ | \ |a| = 1\right\}$. At the TX, with $L_\textrm{T}$ RF chains, the analog precoder matrix is denoted as $\mathbf{F}_\textrm{RF} \in \mathbb{C}^{N_{\textrm{T}} \times L_\textrm{T}}$ and the baseband precoder matrix is denoted as $\mathbf{F}_\textrm{BB} \in \mathbb{C}^{L_\textrm{T} \times N_\textrm{s}}$. The received signal $\textbf{y} \in \mathbb{C}^{N_\textrm{R} \times 1}$ can be expressed as:
 \begin{equation}\label{eq:receiver_system}
 \mathbf{y} = \mathbf{H}\mathbf{F}_\textrm{RF}\mathbf{F}_\textrm{BB}\mathbf{x} + \mathbf{n},
 \end{equation}
 where $\mathbf{x} \in \mathbb{C}^{N_\textrm{s} \times 1}$ is the transmit symbol vector and $\mathbf{n} \in \mathbb{C}^{N_\textrm{R} \times 1}$ is a noise vector with independent and identically distributed entries and follow the complex Gaussian distribution with zero mean and $\sigma_\textrm{n}^2$ variance, i.e., $\mathbf{n} \sim \mathcal{C}\mathcal{N}(\mathbf{0},\sigma_\textrm{n}^2 \mathbf{I}_{N_\textrm{R}})$.
 
As widely used in the existing literature, we consider the linear additive quantization noise model (AQNM) to represent the distortion of quantization \cite{orhanITA2015}.
Given that $Q(\cdot)$ denotes a uniform scalar quantizer then for the scalar complex input $x \in \mathbb{C}$ that is applied to both the real and imaginary parts, we have that,
\begin{equation}\label{eq:approx}
Q(x) \approx \delta x + \epsilon,
\end{equation}
where $\delta = \sqrt{1-\frac{\pi \sqrt{3}}{2}2^{-2 b}} \in [m,M]$ is the multiplicative distortion parameter for a bit resolution equal to $b$ \cite{mezghaniISIT2012} where $m$ and $M$ denote the minimum and maximum value of the range. Note that the introduced error in the linear approximation in \eqref{eq:approx} decreases for larger resolutions. However, our proposed solution focuses on EE maximization and this linear approximation does not impact the performance significantly as observed from the simulation results in Section IV.
The parameter $\epsilon$ is the additive quantization noise with $\epsilon \sim \mathcal{CN}(0, \sigma_{\epsilon}^2)$ , where  $\sigma_{\epsilon} = \sqrt{1-\frac{\pi \sqrt{3}}{2}2^{-2 b}} \sqrt{\frac{\pi \sqrt{3}}{2}2^{-2 b}}$. Based on AQNM, the vector containing the complex output of all the ADCs can be expressed as follows:
\begin{equation}\label{eq:quantization_mimo_model}
Q(\mathbf{W}_\textrm{RF}^H\mathbf{y} ) \approx \boldsymbol{\Delta}^H \mathbf{W}_\textrm{RF}^H \mathbf{y} + \boldsymbol{\epsilon}, 
\end{equation}
where $Q(\mathbf{W}_\textrm{RF}^H \mathbf{y}) \in \mathbb{C}^{L_{\textrm{R}} \times 1}$ and $\boldsymbol{\Delta} = \boldsymbol{\Delta}^H \in \mathbb{C}^{L_\textrm{R} \times L_\textrm{R}}$ is a diagonal matrix with values depending on the ADC resolution $b_i$ of each ADC. Specifically, each diagonal entry of $\mathbf{\Delta}$ is given by: 
\begin{equation}\label{eq:delta_distorsion_rf}
[\boldsymbol{\Delta}]_{ii} = \sqrt{1-\frac{\pi \sqrt{3}}{2}2^{-2 b_i}} \in [m,M] \hspace{1mm} \forall \hspace{1mm} i=1,\ldots, L_{\textrm{R}},
\end{equation}
where, for simplicity, we assume that the range $[m,M]$ is the same for each one of the ADCs. The second term of \eqref{eq:quantization_mimo_model} expresses the additive quantization noise for all RF chains, with $\boldsymbol{\epsilon} \in \mathcal{C}\mathcal{N}(\mathbf{0}, \mathbf{C}_{\epsilon})$ \cite{aryanICC2019} 
where $\mathbf{C}_\epsilon$ is a diagonal covariance matrix with entries as follows: 
\begin{equation}
    [\mathbf{C}_{\epsilon}]_{ii} = \left(1-\frac{\pi \sqrt{3}}{2}2^{-2 b_i}\right) \left(\frac{\pi \sqrt{3}}{2}2^{-2 b_i}\right) \forall \hspace{1mm} i=1,\ldots, L_{\textrm{R}}.
\end{equation}



\begin{figure}[t]
 	\begin{center}
 		\includegraphics[width=0.49\textwidth, trim=225 120 40 130,clip ]{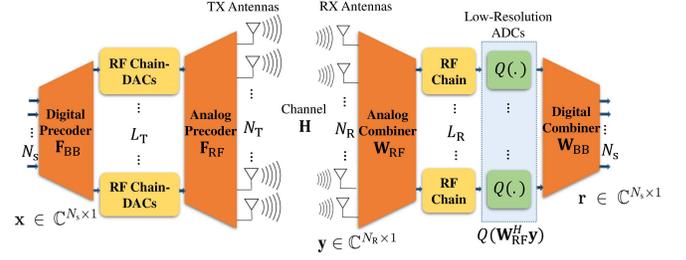}
 		\caption{A mmWave A/D hybrid MIMO system with low resolution ADCs.}
 	\end{center}
 	\vspace{-1mm}
 \end{figure}



After the effect of the quantization and application of the baseband combining matrix, the output $\mathbf{r} \in \mathbb{C}^{N_\textrm{s} \times 1}$ at the RX can be expressed as:
\begin{align}\label{eq:receiver_output}
\mathbf{r} =  
\mathbf{W}_\textrm{BB}^H \mathbf{\Delta}^H \mathbf{W}_\textrm{RF}^H \mathbf{y} + \mathbf{W}_\textrm{BB}^H \boldsymbol{\epsilon}. 
\end{align}
Based on the received signal expression in \eqref{eq:receiver_system}, we can express \eqref{eq:receiver_output} as follows:
\begin{align}\label{eq:system_model}
\mathbf{r} =  
\mathbf{W}_\textrm{BB}^H \mathbf{\Delta}^H \mathbf{W}_\textrm{RF}^H \mathbf{H}\mathbf{F}_\textrm{RF}\mathbf{F}_\textrm{BB}\mathbf{x} + \underbrace{\mathbf{W}_\textrm{BB}^H \mathbf{\Delta}^H \mathbf{W}_\textrm{RF}^H \mathbf{n} + \mathbf{W}_\textrm{BB}^H \boldsymbol{\epsilon}}_{\boldsymbol{\eta}},
\end{align}
where $\boldsymbol{\eta}$ is the combined effect of the Gaussian and the quantization noise with $\boldsymbol{\eta} \sim \mathcal{CN}(\mathbf{0}, \mathbf{R}_\eta)$. Here $\mathbf{R}_\eta \in \mathbb{C}^{L_\textrm{R}\times L_\textrm{R}}$ is the combined noise covariance matrix with, 
\begin{equation}\label{eq:R_eta}
\mathbf{R}_\eta = \sigma_\textrm{n}^2 \mathbf{W}_\textrm{BB}^H\mathbf{\Delta}^H\mathbf{W}_\textrm{RF}^H\mathbf{W}_\textrm{RF}\mathbf{\Delta}\mathbf{W}_\textrm{BB} + \mathbf{W}_\textrm{BB}^H\mathbf{C}_\epsilon\mathbf{W}_\textrm{BB}.
\end{equation}

\section{Bit Allocation and Hybrid Combiner Design}
\subsection{Problem Formulation}
Let us consider a point-to-point MIMO system with the linear quantization model. We define the EE as the ratio of the information rate and the total consumed power as,
\begin{equation}\label{eq:ee_problem}
 EE(\mathbf{W}_\textrm{RF}, \mathbf{\Delta}, \mathbf{W}_\textrm{BB}) \triangleq \frac{R(\mathbf{W}_\textrm{RF}, \mathbf{\Delta}, \mathbf{W}_\textrm{BB})}{P(\mathbf{\Delta})} \,\, \textrm{(bits/Joule)},
\end{equation}
where the information rate is defined as,
\begin{align}\label{eq:rate}
R(\mathbf{W}_\textrm{RF}, \mathbf{\Delta}, \mathbf{W}_\textrm{BB}) \triangleq \log_2 \vert \mathbf{I}_{L_\textrm{R}} + \frac{\mathbf{R}^{-1}_\eta}{N_\textrm{s}} \mathbf{W}_\textrm{BB}^H \mathbf{\Delta}^H\mathbf{W}_\textrm{RF}^H\mathbf{H} \mathbf{F} \times \nonumber \\ \mathbf{F}^H\mathbf{H}^H \mathbf{W}_\textrm{RF} \mathbf{\Delta} \mathbf{W}_\textrm{BB} \vert \hspace{1mm} \textrm{(bits/s)},
\end{align}
where the A/D hybrid precoder $\mathbf{F}=\mathbf{F}_\textrm{RF}\mathbf{F}_\textrm{BB}\in \mathbb{C}^{N_\textrm{T} \times N_\textrm{s}}$. 


Similar to the power model at the TX in \cite{aryanICC2019}, the total consumed power at the RX is expressed as:
\begin{equation}\label{eq:power_total}
P(\mathbf{\Delta}) = P_\textrm{D} + N_\textrm{R}P_\textrm{R} +  N_\textrm{R}L_\textrm{R}P_\textrm{PS} + P_\textrm{CP} \,\, \textrm{(W)},
\end{equation}
where $P_\textrm{PS}$ is the power per phase shifter, $P_\textrm{R}$ is the power per antenna, $P_\textrm{D}$ is the power associated with the total quantization operation, and following \eqref{eq:delta_distorsion_rf} and \cite{orhanITA2015}, we have 
\begin{equation}\label{eq:power_adc}
P_\textrm{D} \!=\! P_\textrm{ADC} \sum_{i=1}^{L_{\textrm{R}}} 2^{b_i} \!=\! P_\textrm{ADC} \sum_{i=1}^{L_{\textrm{R}}} \left( \frac{\pi \sqrt{3}}{2 (1 - [\boldsymbol{\Delta}]_{ii}^2)} \right)^{\frac{1}{2}} \,\, \textrm{(W)},
\end{equation}
where $P_\textrm{ADC}$ is the power consumed per bit in the ADC and $P_\textrm{CP}$ is the power required by all circuit components. 

Considering the rate and power model in \eqref{eq:rate} and \eqref{eq:power_total}, respectively, we can express the following fractional problem:
\begin{align}
(\mathcal{P}_\textrm{1}): \hspace{5pt} \max_{\mathbf{W}_\textrm{RF}, \mathbf{\Delta}, \mathbf{W}_\textrm{BB}} \,\, \frac{R(\mathbf{W}_\textrm{RF}, \mathbf{\Delta}, \mathbf{W}_\textrm{BB})}{P(\mathbf{\Delta})} \nonumber \\
\textrm{ subject to } \mathbf{W}_\textrm{RF} \in \mathcal{W}^{N_\textrm{R} \times L_\textrm{R}}, 
\mathbf{\Delta} \in \mathcal{D}^{L_\textrm{R} \times L_\textrm{R}} \nonumber,
\end{align}
where the set $\mathcal{D}$ represents the finite states of the quantizer and is defined as,
\[\mathcal{D} = \left\{\mathbf{\Delta} \in \mathbb{R}^{L_\textrm{R} \times L_\textrm{R}} \big| m \le [\mathbf{\Delta}]_{ii} \le M \hspace{1mm} \forall \hspace{1mm} i = 1,..., L_\textrm{R}\right\}.\]
The channel's singular value decomposition (SVD) is written as $\mathbf{H} = \mathbf{U}_\textrm{H} \mathbf{\Sigma}_\textrm{H} \mathbf{V}_\textrm{H}^H$, where $\mathbf{U}_\textrm{H} \in \mathbb{C}^{N_\textrm{R} \times N_\textrm{R}}$ and $\mathbf{V}_\textrm{H} \in \mathbb{C}^{N_\textrm{T} \times N_\textrm{T}}$ are unitary matrices, and $\mathbf{\Sigma}_\textrm{H} \in \mathbb{R}^{{N_\textrm{R} \times N_\textrm{T}}}$ is a rectangular matrix of singular values in decreasing order whose diagonal elements are non-negative real numbers and whose non-diagonal elements are zero. The optimal, fully digital combiner matrix $\mathbf{W}_\textrm{opt}$ consists of the $N_\textrm{s}$ columns of the left singular matrix $\mathbf{U}_\textrm{H}$. Our goal, by solving ($\mathcal{P}_\textrm{1}$), is to obtain the combiner matrices and the bit resolution matrix in an optimal manner. We introduce the novel decomposition of the A/D hybrid combiner to three parts representing the analog combiner matrix, the bit resolution matrix and digital combiner matrix, 
i.e., $\mathbf{W}_\textrm{RF}\mathbf{\Delta} \mathbf{W}_\textrm{BB}$. So the Euclidean distance $\norm{\mathbf{W}_\textrm{opt} - \mathbf{W}_\textrm{RF}\mathbf{\Delta}\mathbf{W}_\textrm{BB}}_F^2$ should be as small as possible for a maximum throughput combiner design. Note that we optimize over the bit resolution matrix with varying resolutions and the choice of combiner matrices at the RX.

\begin{proposition}
The maximization of the fractional problem $(\mathcal{P}_\textrm{1})$ is equivalent with the solution of the following problem:
\begin{align*}
\hspace{-5pt}(\mathcal{P}_\textrm{2}): \hspace{5pt} &\min_{\mathbf{W}_\textrm{RF},\mathbf{\Delta},\mathbf{W}_\textrm{BB}} \frac{1}{2}\norm{\mathbf{W}_\textrm{opt}-\mathbf{W}_\textrm{RF}\mathbf{\Delta}\mathbf{W}_\textrm{BB}}_F^2 + \gamma P(\mathbf{\Delta}),\nonumber \\
& \textrm{subject to}
\hspace{1mm} \mathbf{W}_\textrm{RF} \in \mathcal{W}^{N_\textrm{R} \times L_\textrm{R}}, \mathbf{\Delta} \in \mathcal{D}^{L_\textrm{R} \times L_\textrm{R}},
\end{align*} 
where the parameter $\gamma \in \mathbb{R}^+$ denotes the trade-off between the rate and the power consumption.
\end{proposition}
\begin{proof} The main idea to prove the equivalence is first to apply the Dinkelbach approach to transform the fractional problem into an affine one \cite{dinkelbach}. Afterwards, based on \cite{ayachTWC2014, aryanIET2016}, the maximization of the rate $R$ can be expressed as minimization of the Euclidean distance between the computed A/D hybrid combiner and the optimal, fully digital combiner $\mathbf{W}_\textrm{opt}$. The details of this proof are omitted due to space limitations. \end{proof}

Parameter $\gamma$ also determines how close is the solution of $(\mathcal{P}_\textrm{2})$ to $(\mathcal{P}_\textrm{1})$. In this work, $\gamma$ is selected after an exhaustive search over all the possible values in the range of [0.001, 0.1] and the value which gives the best result for ($\mathcal{P}_\textrm{2}$) is selected. Problem $(\mathcal{P}_\textrm{2})$ is non-convex due to the constraints on the structure of matrix $\mathbf{W}_\textrm{RF}$. Similar non-convex problems have been recently addressed in the literature via alternating direction method of multipliers (ADMM) based solutions \cite{ADMM,TSINOSNC1,TSINOSNC2}. 

\subsection{Proposed ADMM Solution}
In the following we develop an iterative procedure for solving $(\mathcal{P}_\textrm{2})$ based on the ADMM approach \cite{ADMM}. This method, is a variant of the standard augmented Lagrangian method that uses partial updates (similar to the Gauss-Seidel method for the solution of linear equations) to solve constrained optimization problems. This method replaces a constrained minimization problem by a series of unconstrained problems and add a penalty term to the objective function. This penalty improves robustness compared to other optimization methods for constrained problems (for example, the dual ascent method) and in particular achieves convergence without the need of specific assumptions for the objective function, i.e., strict convexity and finiteness. 
The interested reader may refer to \cite{ADMM} for further information. 

We first transform $(\mathcal{P}_\textrm{2})$ into a form that can be addressed via ADMM. By using the auxiliary variable $\mathbf{Z}$, $(\mathcal{P}_\textrm{2})$ can be written in the following form:
\begin{align*}
\hspace{-2pt}(\mathcal{P}_\textrm{3}): \hspace{2pt}  &\min_{\substack{\mathbf{Z},\mathbf{W}_\textrm{RF}, \mathbf{\Delta},\mathbf{W}_\textrm{BB}}} \frac{1}{2}\Vert\mathbf{W}_\textrm{opt}-\mathbf{Z}\Vert_F^2 + \mathds{1}_{\mathcal{W}^{N_\textrm{R} \times L_\textrm{R}}}\{\mathbf{W}_\textrm{RF}\} \nonumber \\ 
& \hspace{19mm} + \mathds{1}_{\mathcal{D}^{L_\textrm{R} \times L_\textrm{R}}}\{\mathbf{\Delta}\} + \gamma P(\mathbf{\Delta}), \\
& \hspace{20mm} \textrm{subject to} \hspace{1mm} \mathbf{Z} = \mathbf{W}_\textrm{RF}\mathbf{\Delta}\mathbf{W}_\textrm{BB}.
\end{align*} 
%
%
%
Problem $(\mathcal{P}_\textrm{3})$ formulates the A/D hybrid combiner matrix design as a matrix factorization problem. That is, the overall combiner $\mathbf{Z}$ is sought so that it minimizes the Euclidean distance to the optimal, fully digital combiner $\mathbf{W}_\textrm{opt}$ while supporting decomposition into three factors: the analog combiner matrix $\mathbf{W}_\textrm{RF}$, the matrix $\boldsymbol{\Delta}$ which is related to the resolution of each ADC and the digital combiner matrix $\mathbf{W}_\textrm{BB}$.

The augmented Lagrangian function of $(\mathcal{P}_\textrm{3})$ is given by, 
\begin{align}
\mathcal{L}(\mathbf{Z},\!\mathbf{W}_\textrm{RF},\!\mathbf{\Delta},\!\mathbf{W}_\textrm{BB},\!\mathbf{\Lambda}) \! = \! \frac{1}{2}\lVert\mathbf{W}_\textrm{opt}\!-\!\mathbf{Z}\rVert^2_F  \!+\! \mathds{1}_{\mathcal{W}^{N_\textrm{R} \times L_\textrm{R}}}\{\mathbf{W}_\textrm{RF}\} \nonumber\\ 
\!+\mathds{1}_{\mathcal{D}^{L_\textrm{R} \times L_\textrm{R}}}\{{\mathbf{\Delta}}\}  \!+\! \frac{\alpha}{2}\lVert\mathbf{Z}\!+\!\mathbf{\Lambda}/\alpha\! - \!\mathbf{W}_\textrm{RF}\mathbf{\Delta}\mathbf{W}_\textrm{BB}\rVert_F^2 \!+\! \gamma P(\mathbf{\Delta}),
\label{EQ:AUGLAN} 
\end{align}
where $\alpha$ is a scalar penalty parameter and $\mathbf{\Lambda} \in \mathbb{C}^{N_\textrm{R}\times L_\textrm{R}}$ is the Lagrange Multiplier matrix. 
%
According to ADMM \cite{ADMM}, the solution to  $(\mathcal{P}_\textrm{3})$ is derived by the following iterative steps:
\allowdisplaybreaks
\begin{align}
&\hspace{0pt}(\mathcal{P}_\textrm{3A}): \hspace{0pt} \mathbf{Z}_{(n)} = \arg \min_{\mathbf{Z}}  \frac{1}{2}\lVert(1+\alpha)\mathbf{Z}-\mathbf{W}_\textrm{opt}+\mathbf{\Lambda}_{(n-1)} \nonumber \\
& \hspace{22mm} -\alpha\mathbf{W}_{\textrm{RF}(n-1)}\mathbf{\Delta}_{(n-1)}\mathbf{W}_{\textrm{BB}(n-1)}\rVert_F^2, \nonumber\\
%
&\hspace{0pt}(\mathcal{P}_\textrm{3B}): \hspace{0pt} \mathbf{W}_{\textrm{RF}(n)} = \arg \min_{\mathbf{W}_\textrm{RF}}  \mathds{1}_{\mathcal{W}^{N_\textrm{R} \times L_\textrm{R}}}\{\mathbf{W}_\textrm{RF}\} \!+\! \frac{\alpha}{2}\times \nonumber \\  
& \hspace{10mm} \norm{\mathbf{Z}_{(n)}+\mathbf{\Lambda}_{(n-1)}/\alpha- \mathbf{W}_\textrm{RF}\mathbf{\Delta}_{(n-1)}\mathbf{W}_{\textrm{BB}(n-1)}}_F^2, \nonumber\\
%
&\hspace{0pt}(\mathcal{P}_\textrm{3C}): \hspace{0pt} \mathbf{\Delta}_{(n)} = \arg \min_{\mathbf{\Delta}}\lVert\mathbf{y}_\textrm{c}-\mathbf{\Psi}\textrm{vec}(\mathbf{\Delta})\rVert_2^2 +\gamma P(\mathbf{\Delta}), \nonumber \\
& \hspace{23mm} \textrm{ subject to } \mathbf{\Delta} \in \mathcal{D},\nonumber\\
%
&\hspace{0pt}(\mathcal{P}_\textrm{3D}): \hspace{0pt} \mathbf{W}_{\textrm{BB}(n)} = \arg \min_{\mathbf{W}_\textrm{BB}} \frac{\alpha}{2}\lVert\mathbf{Z}_{(n)}+\mathbf{\Lambda}_{(n-1)}/\alpha \nonumber \\ 
& \hspace{27mm} -\mathbf{W}_{\textrm{RF}(n)}\mathbf{\Delta}_{(n)}\mathbf{W}_\textrm{BB}\rVert_F^2, \nonumber\\
%
&\mathbf{\Lambda}_{(n)} = \mathbf{\Lambda}_{(n-1)} + \alpha\left(\mathbf{Z}_{(n)}-{\mathbf{W}_\textrm{RF}}_{(n)}\mathbf{\Delta}_{(n)}{\mathbf{W}_\textrm{BB}}_{(n)}\right),
\end{align}
where $n$ denotes the iteration index, $\mathbf{y}_\textrm{c}\!=\! \textrm{vec}(\textbf{Z}_{(n)}\!+\! \mathbf{\Lambda}_{(n-1)}/\alpha)$ and $\mathbf{\Psi}\!=\! {\mathbf{W}_\textrm{BB}}_{(n-1)}\! \otimes \! {\mathbf{W}_\textrm{RF}}_{(n)}$ ($\otimes$ is the Khatri-Rao product). 

We solve the optimization problems $(\mathcal{P}_\textrm{3A})$-$(\mathcal{P}_\textrm{3D})$ and the solutions are provided in Algorithm 1. The algorithm provides the complete procedure to obtain the optimal analog combiner matrix $\mathbf{W}_\textrm{RF}$, the optimal bit resolution matrix $\mathbf{\Delta}$ and the optimal baseband (or digital) combiner matrix $\mathbf{W}_\textrm{BB}$. It starts by initializing the entries of the matrices $\mathbf{Z}$, ${\mathbf{W}_\textrm{RF}}$, $\mathbf{\Delta}$, ${\mathbf{W}_\textrm{BB}}$ with random values and the entries of the Lagrange multiplier matrix $\mathbf{\Lambda}$ with zeros. For iteration index $n$, $\mathbf{Z}_{(n)}$, ${\mathbf{W}_\textrm{RF}}_{(n)}$, $\mathbf{\Delta}_{(n)}$ and ${\mathbf{W}_\textrm{BB}}_{(n)}$ are updated at each iteration step using the solutions provided in Steps 4, 7, 8, 10 and 11 of Algorithm 1. In Step 7, ${\Pi}_{\mathcal{W}}$ is the operator that projects the solution onto the set $\mathcal{W}$. This is computed by solving the following optimization problem \cite{AGL}:
\begin{align}
\hspace{-25pt}(\mathcal{P}_\textrm{4}): \hspace{25pt} &\min_{\mathbf{A}_{\mathcal{W}}}\|\mathbf{A}_{\mathcal{W}}-\mathbf{A}\|_F^2, \textrm{subject to} \ \mathbf{A}_{\mathcal{W}} \in \mathcal{W}, \nonumber
\end{align}
where $\mathbf{A}$ is an arbitrary matrix and $\mathbf{A}_{\mathcal{W}}$ is its projection onto the set $\mathcal{W}$. The solution to $(\mathcal{P}_\textrm{4})$ is given by the phase of the complex elements of $\mathbf{A}$. Thus, for $\mathbf{A}_\mathcal{W}  = \Pi_\mathcal{W}\{\mathbf{A}\}$ we have
\begin{align}
\mathbf{A}_\mathcal{W}(x,y) = \begin{cases}
0, \ &\mathbf{A}(x,y) = 0 \\
\frac{\mathbf{A}(x,y)}{\left|\mathbf{A}(x,y)\right|}, \ &\mathbf{A}(x,y) \neq 0 
\end{cases},  
\label{EQ:PROJ_F}
\end{align}
where $\mathbf{A}_\mathcal{W}(x,y)$ and $\mathbf{A}(x,y)$ are the elements at the $x$th row-$y$th column of matrices $\mathbf{A}_\mathcal{W}$ and $\mathbf{A}$, respectively. Furthermore, as shown in Step 8, the minimization problem in ($\mathcal{P}_\textrm{3C}$) is solved by implementing CVX \cite{cvx}. A termination criterion related to the maximum permitted number of iterations of the ADMM sequence ($N_\textrm{max}$) is considered. Upon convergence, the number of bits for each ADC is obtained by using \eqref{eq:delta_distorsion_rf} and quantized to the nearest integer value. 

\subsubsection*{Computational complexity analysis of Algorithm 1}
In Algorithm 1, mainly Step 8 involves multiplication by $\mathbf{\Psi}$ whose dimensions are $L_\textrm{R}N_\textrm{R}\times N_\textrm{s} L_\textrm{R}$. In general, the solution of $(\mathcal{P}_\textrm{3C})$ can be upper-bounded by $\mathcal{O}((L_\textrm{R}^2N_\textrm{R}N_\textrm{s})^3)$ which can be improved significantly by exploiting the structure of $\mathbf{\Psi}$.
\begin{algorithm}[t]
    \caption{Proposed ADMM Solution for the A/D Hybrid Combiner Design}
  \begin{algorithmic}[1]
    \STATE \textbf{Initialize:} $\mathbf{Z}$, ${\mathbf{W}_\textrm{RF}}$, $\mathbf{\Delta}$, ${\mathbf{W}_\textrm{BB}}$ with random values, $\mathbf{\Lambda}$ with zeros, $\alpha=1$ and $n=1$
    \WHILE{$n \leq N_\textrm{max}$}
    \STATE $\mathbf{A} = \alpha {\mathbf{W}_\textrm{RF}}_{(n-1)}\mathbf{\Delta}_{(n-1)} {\mathbf{W}_\textrm{BB}}_{(n-1)}$. 
      	\STATE $\mathbf{Z}_{(n)} = \frac{1}{\alpha + 1} \big(\mathbf{W}_\textrm{opt} - \mathbf{\Lambda}_{(n-1)} +  \mathbf{A}\big)$.
      	\STATE $\mathbf{B} = \mathbf{\Lambda}_{(n-1)}+\alpha\mathbf{Z}_{(n)}$.
      	\STATE $\mathbf{C} = \alpha\mathbf{\Delta}_{(n-1)}{\mathbf{W}_\textrm{BB}}_{(n-1)}{\mathbf{W}_\textrm{BB}}_{(n-1)}^H\mathbf{\Delta}_{(n-1)}^H$.
\STATE ${\mathbf{W}_\textrm{RF}}_{(n)} =  \Pi_\mathcal{W}\{\mathbf{B}{\mathbf{W}_\textrm{BB}}_{(n-1)}^H\mathbf{\Delta}_{(n-1)}^H\mathbf{C}^{-1}\}$. 
\STATE Update $\mathbf{\Delta}_{(n)}$ by solving ($\mathcal{P}_\textrm{3C}$) using CVX \cite{cvx}.
\STATE $\mathbf{D} = \alpha\mathbf{\Delta}_{(n)}^H{\mathbf{W}_\textrm{RF}}_{(n)}^H{\mathbf{W}_\textrm{RF}}_{(n)}\mathbf{\Delta}_{(n)}$.
\STATE ${\mathbf{W}_\textrm{BB}}_{(n)} = \mathbf{D}^{-1}\mathbf{\Delta}_{(n)}^H {\mathbf{W}_\textrm{RF}}_{(n)}^H\mathbf{B}$.
\STATE $\mathbf{\Lambda}_{(n)} = \mathbf{\Lambda}_{(n-1)} + \alpha\left(\mathbf{Z}_{(n)}-{\mathbf{W}_\textrm{RF}}_{(n)}\mathbf{\Delta}_{(n)}{\mathbf{W}_\textrm{BB}}_{(n)}\right)$.
 \STATE $n \gets n+1$
      \ENDWHILE
    \RETURN ${\mathbf{W}_\textrm{RF}}_{(N_\textrm{max})}$, $\mathbf{\Delta}_{(N_\textrm{max})}$, ${\mathbf{W}_\textrm{BB}}_{(N_\textrm{max})}$
  \end{algorithmic}
\end{algorithm}

\section{Simulation Results}
In this section, we evaluate the performance of the proposed ADMM technique using computer simulation results. 
The results have been averaged over 1,000 Monte-Carlo realizations.
\subsubsection*{System setup}
We set the following parameters, unless specified otherwise, to obtain the desired results: $N_\textrm{T}$ = $32$, $N_\textrm{R}$ = $16$, $L_\textrm{R} = 4$, $N_\textrm{s} = 4$, $N_\textrm{cl}=2$, $N_\textrm{ray}=4$, $N_\textrm{max}=40$, $m=1$, $M=8$, $\alpha=1$ and $\sigma^2_{\alpha, i}$ = $1$. The azimuth angles of departure and arrival are computed with uniformly distributed mean angles; each cluster follows a Laplacian distribution about the mean angle. The antenna elements in the ULA are spaced by distance $d$ = $\lambda/2$. The signal-to-noise ratio (SNR) is given by the inverse of the noise variance, i.e., $1/\sigma_\textrm{n}^2$. The transmit vector $\mathbf{x}$ is composed of the normalized i.i.d. Gaussian symbols. The values used for the terms in the power model in \eqref{eq:power_total} of Section III are $P_\textrm{ADC} = 100$ mW, $P_\textrm{CP} = 10$ W, $P_\textrm{R} = 100$ mW and $P_\textrm{PS} = 10$ mW. Note that to measure the spectral efficiency (SE) performance, we compute the ratio $R/B$ bits/s/Hz  where $B$ represents the bandwidth, and for the simulations we set $B=1$ Hz. For simulations, the precoder matrix $\mathbf{F}$ is considered equal to the optimal fully digital precoder matrix \cite{ayachTWC2014, aryanIET2016}, i.e., the product of  $1/\sqrt{N_\textrm{s}}$ and first $N_\textrm{s}$ columns of the right singular matrix $\mathbf{V}_\textrm{H}$. 
 
\subsubsection*{Convergence of the proposed ADMM solution}
Fig. 2 shows the convergence of the ADMM solution as proposed in Algorithm 1 to obtain the optimal bit resolution at each ADC and corresponding optimal combiner matrices. The proposed solution converges rapidly at around 20 iterations and mean square error (MSE), $\norm{\mathbf{W}_\textrm{opt} - \mathbf{W}_{\textrm{RF}(N_\textrm{max})}\mathbf{\Delta}_{(N_\textrm{max})}\mathbf{W}_{\textrm{BB}(N_\textrm{max})}}_F^2$, goes as low as -20 dB. A lower number of RX antennas shows lower MSE as expected, since fewer parameters are required to be estimated.

\begin{figure}[t]
	\begin{center}
		\includegraphics[width=0.51\textwidth]{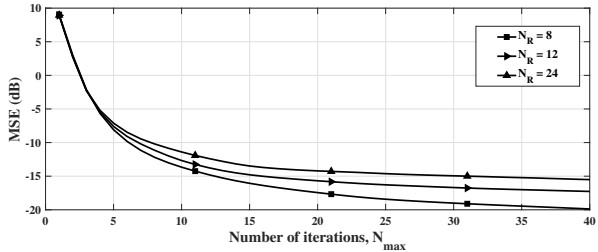}
	\caption{Convergence of the ADMM solution for different $N_\textrm{R}$ at $\gamma=0.01$.}
        \label{fig:conv}
	\end{center}
	\vspace{-2mm}
\end{figure}

\begin{figure}[t]
 \centering \includegraphics[width=0.515\textwidth]{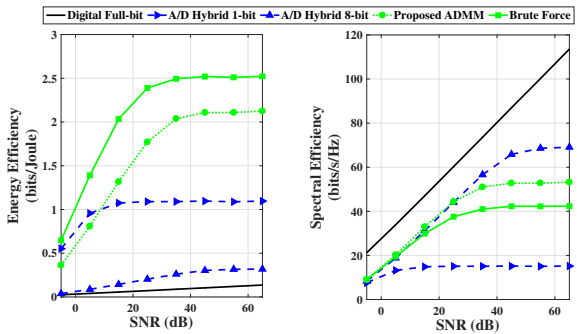}
 		\caption{EE and SE performance w.r.t. SNR at  $N_\textrm{R}=16$ and $\gamma=0.01$.}
 		\vspace{-5mm}
 \end{figure}
 
 \begin{figure}[t]  
 \centering \includegraphics[width=0.51\textwidth]{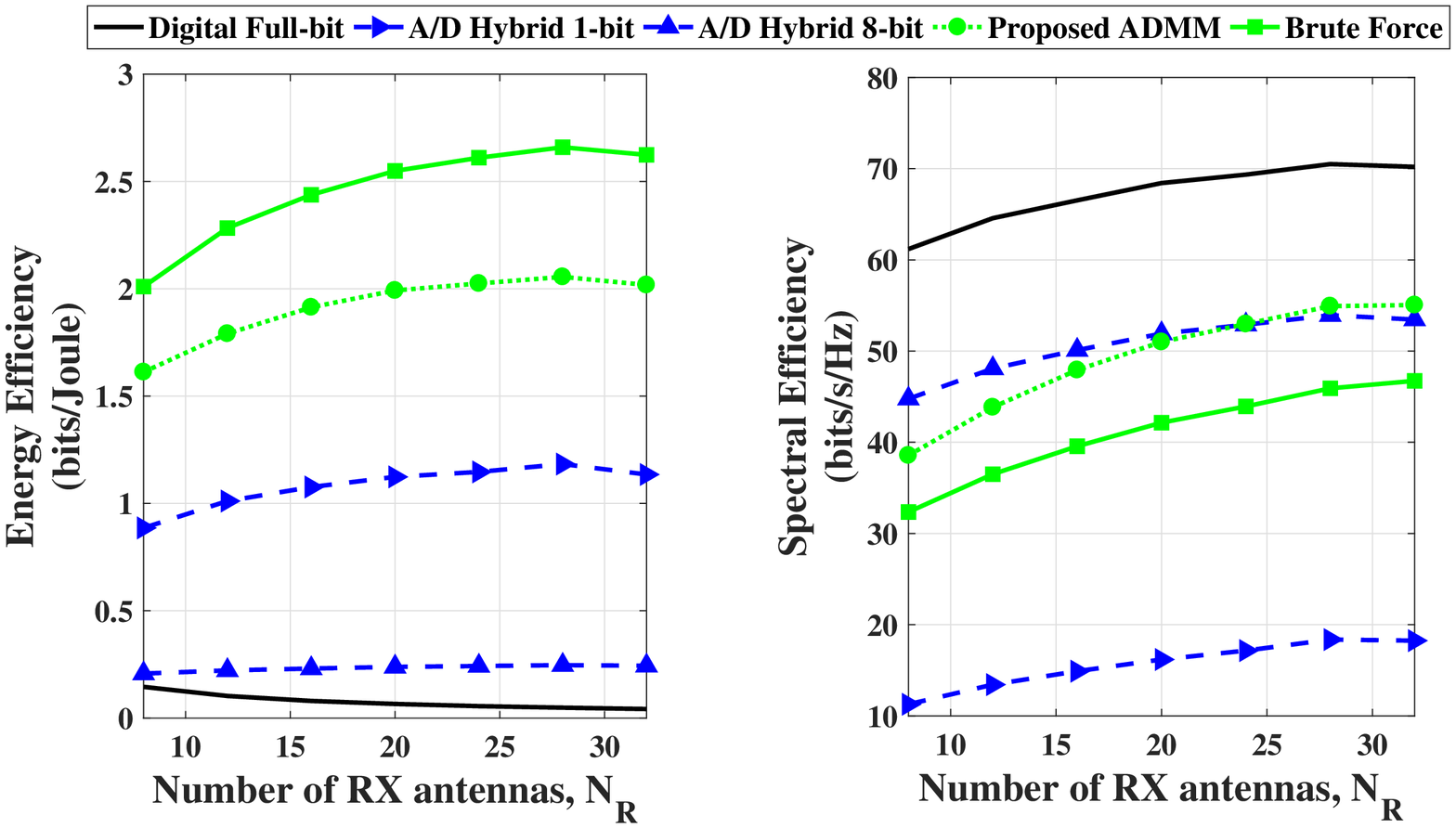}
 		\caption{EE and SE performance w.r.t. $N_\textrm{R}$ at SNR = $30$ dB and $\gamma=0.01$.}
 		\vspace{-2mm}
 \end{figure}
 
\begin{figure}[t] 
\centering \includegraphics[width=0.515\textwidth]{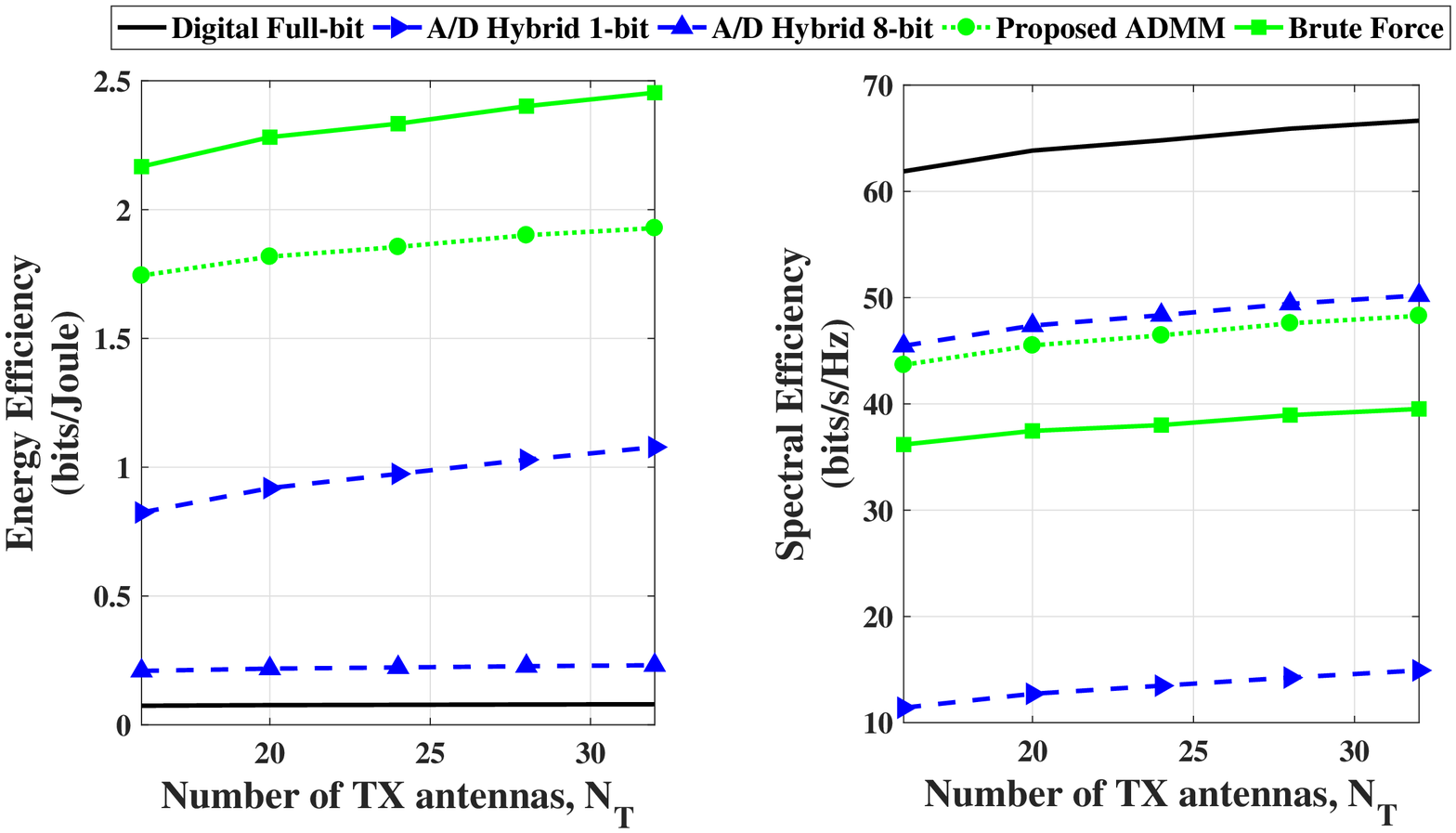}
 		\caption{EE and SE performance w.r.t. $N_\textrm{T}$ at SNR = $30$ dB and $\gamma=0.01$.}
 		\vspace{-5mm}
 \end{figure}
 
\begin{figure}[t] 
\centering \includegraphics[width=0.495\textwidth]{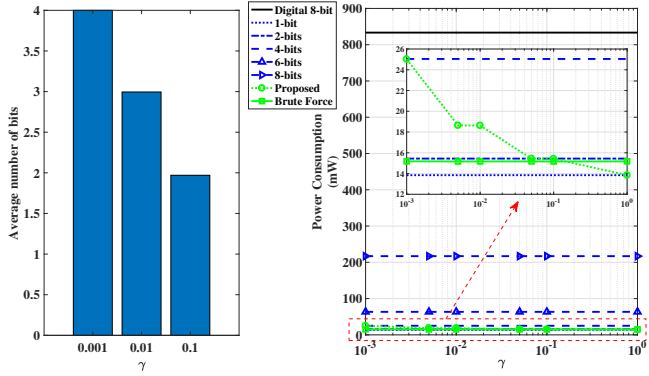}
 		\caption{Average number of bits for proposed ADMM and power consumption w.r.t. $\gamma$ at SNR = $30$ dB.}
 			\vspace{-2mm}
 \end{figure} 
 
 \begin{figure}[t]  
\centering \includegraphics[width=0.515\textwidth]{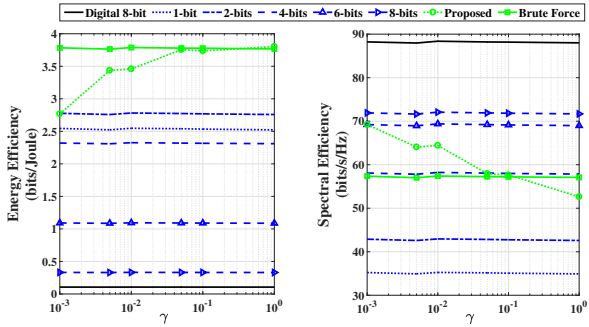}
 		\caption{EE and SE performance w.r.t. $\gamma$ at SNR = $30$ dB.}
 		\vspace{-5mm}
 \end{figure}  

\subsubsection*{Benchmark techniques}
\subsubsection{Digital combining with full-bit resolution}
We consider the conventional fully digital beamforming architecture, where the number of RF chains at the RX is equal to the number of RX antennas, i.e., $L_\textrm{R} = N_\textrm{R}$. The fully digital combining solution may be provided by SVD and waterfilling \cite{rappa2}. In terms of the resolution sampling, we consider full-bit resolution, i.e., $M = 8$-bit, which represents the optimum from the achievable SE perspective.

\subsubsection{A/D Hybrid combining with 1-bit and 8-bit resolutions}
We also consider a A/D hybrid combining architecture with $L_\textrm{R} < N_\textrm{R}$, for two cases of bit resolution: a) 1-bit resolution which usually shows reasonable EE performance, and b) 8-bit resolution which usually shows high SE results. 

\subsubsection{Brute force with A/D hybrid combining}
We also implement an exhaustive search approach as an upper bound for EE maximization called brute force (BF), based on \cite{ranziSAC2016}, which clearly shows the energy-rate performance trade-offs in the simulations. It makes a search over the number of RF chains $L_\textrm{R}$ and all the available bit resolutions, i.e., $b = 1,..., M$. It then finds the best EE out of all the possible cases and chooses the corresponding optimal resolution for each ADC. This method provides the best possible EE performance, but it is computationally intractable for $L_\textrm{R}>4$.

Fig. 3 shows the performance of the proposed ADMM solution compared with existing benchmark techniques with respect to (w.r.t.) SNR at $N_\textrm{R}=16$. The proposed ADMM solution achieves high EE which has performance close to the BF approach and better than the 8-bit hybrid, 1-bit hybrid and full-bit digital baselines. For example, at SNR = $20$ dB, the proposed ADMM solution outperforms 1-bit hybrid, 8-bit hybrid and full-bit digital baselines by about 0.45 bits/Joule, 1.375 bits/Joule and 1.44 bits/Joule, respectively. It also exhibits better SE than 1-bit hybrid and has similar performance to the 8-bit hybrid baseline. 

There is an energy-rate trade-off between the proposed solution and the BF approach as we can achieve better rate with lower EE and vice-versa. Moreover, the proposed solution has lower complexity than the BF approach because the BF involves a search over all the possible bit resolutions while the proposed solution directly optimizes the number of bits to obtain an optimal number of bits at each ADC. We constrain the number of RF chains $L_\textrm{R}=4$ for the BF approach due to the high complexity order which is $\mathcal{O}(M^{L_\textrm{R}})$. Also note that the proposed approach enables the selection of different resolutions for different ADCs and thus, it offers a better trade-off for EE versus SE than existing approaches which are based on a fixed ADC resolution. 
 
Figs. 4 and 5 show the performance results w.r.t. the number of RX and TX antennas at 30 dB SNR. The proposed ADMM solution again achieves high EE and performs close to the BF approach and better than the 8-bit hybrid, 1-bit hybrid and full-bit digital baselines. For example, at $N_\textrm{R} = 20$, the proposed ADMM solution outperforms 1-bit hybrid, 8-bit hybrid and full-bit digital baselines by about 0.85 bits/Joule, 1.75 bits/Joule and 1.875 bits/Joule, respectively. Also, for $N_\textrm{T} = 20$, the proposed solution outperforms 1-bit hybrid, 8-bit hybrid and full-bit digital baselines by about 1.0 bits/Joule, 1.5 bits/Joule and 1.625 bits/Joule, respectively. The proposed solution also exhibits better SE than 1-bit hybrid and has similar performance to the 8-bit hybrid baseline. Both the figures follow the energy-rate trade-off with the BF approach.

Furthermore, we investigate the performance over the trade-off parameter $\gamma$ introduced in ($\mathcal{P}_\textrm{2}$). Fig. 6 shows the bar plot of average of the optimal number of bits selected by the proposed solution for each ADC versus $\gamma$. The average optimal number decreases with the increase in $\gamma$, for example, it is 4 for $\gamma = 0.001$, 3 for $\gamma = 0.01$ and 2 for $\gamma = 0.1$. Fig. 6 also shows that the power consumption in the proposed case is considerably low and decreases with the increase in the trade-off parameter $\gamma$ unlike digital 8-bit, several fixed bit hybrid baselines and the BF approach. Fig. 7 shows the EE and SE plots for several solutions w.r.t. $\gamma$. It can be observed that the proposed solution achieves higher EE than the fixed bit allocation solutions and achieves comparable EE and SE results to the BF approach. These curves also show that adjusting $\gamma$ allows the system to vary the energy-rate trade-off.

\section{Conclusion}
This paper proposes an energy efficient mmWave A/D hybrid MIMO system which can vary the ADC bit resolution at the RX. This method uses the decomposition of the A/D hybrid combiner matrix into three parts representing the analog combiner matrix, the bit resolution matrix and the digital combiner matrix. These three matrices are optimized by the novel ADMM solution which outperforms the EE of the full-bit digital, 1-bit hybrid combining and 8-bit hybrid combining baselines. There is an energy-rate trade-off with the BF approach which yields the upper bound for EE maximization. The proposed approach enables the selection of the optimal resolution for each ADC and thus, it offers better trade-off for data rate versus EE than existing approaches based on fixed ADC resolution. In future work, we will jointly optimize the DAC and ADC bit resolution and hybrid precoder and combiner matrices at the TX and the RX.

\bibliographystyle{IEEEtran}

\end{document}